\documentclass{llncs}

\usepackage[utf8]{inputenc}
\usepackage[english]{babel}
\usepackage{enumerate}
\usepackage{stmaryrd}   
\usepackage{xspace}     
\usepackage{amsmath}
\usepackage{amssymb}
\usepackage{paralist}   
\usepackage{shuffle}
\usepackage{wrapfig}
\usepackage{graphicx}

\usepackage{microtype}

\newcommand{\DA}{\textrm{DA}\xspace}
\newcommand{\NDA}{\textrm{NDA}\xspace}
\newcommand{\RA}{\textrm{RA}\xspace}
\newcommand{\sDA}{\textrm{sDA}\xspace}
\newcommand{\sNDA}{\textrm{sNDA}\xspace}
\newcommand{\pDA}{\textrm{pDA}\xspace}
\newcommand{\pNDA}{\textrm{pNDA}\xspace}

\newcommand{\LTL}{\textrm{LTL}\xspace}
\newcommand{\BDLTL}{\textrm{BD-LTL}\xspace}
\newcommand{\pBDLTL}{\textrm{BD-LTL}\ensuremath{^-}\xspace}
\newcommand{\fBDLTL}{\textrm{BD-LTL}\ensuremath{^+}\xspace}
\newcommand{\LRV}{\textrm{LRV}\xspace}

\newcommand{\NDLTL}{\textrm{ND-LTL}\xspace}

\newcommand{\class}[2]{\mathsf{cl}(#1,#2)}
\newcommand{\str}[1]{\mathsf{str}(#1)}
\renewcommand{\vec}[1]{\mathbf{#1}}

\newcommand{\U}{\ensuremath{\operatorname{U}}\xspace}
\renewcommand{\S}{\ensuremath{\operatorname{S}}\xspace}

\newcommand{\X}{\ensuremath{\operatorname{X}}\xspace}
\newcommand{\Y}{\ensuremath{\operatorname{Y}}\xspace}
\newcommand{\G}{\ensuremath{\operatorname{G}}\xspace}
\newcommand{\F}{\ensuremath{\operatorname{F}}\xspace}

\newcommand{\C}{\ensuremath{\operatorname{C}}\xspace}

\newcommand{\cupdot}{\mathrel{\dot{\cup}}}

\title{Ordered Navigation on Multi-attributed Data~Words\thanks{This work is partially supported by EGIDE/DAAD-Procope (LeMon).}\vspace{-0.6em}}

\author{Normann Decker\inst{1} \and
  Peter Habermehl\inst{2}  \and Martin Leucker\inst{1} 
  \and Daniel Thoma\inst{1}\vspace{-0.5em} }

\institute{
  ISP, University of L\"ubeck, Germany\\
  \texttt{\{decker,leucker,thoma\}@isp.uni-luebeck.de}\\[0.5em]
\and
  Univ Paris Diderot, Sorbonne Paris Cit\'e, LIAFA, CNRS, France\\
  \texttt{peter.habermehl@liafa.univ-paris-diderot.fr}
\vspace{-1em}}

\begin{document}

\maketitle

\begin{abstract}
We study temporal logics and automata on multi-attributed data words.
Recently, BD-LTL was introduced as a temporal logic on data words extending LTL by navigation along positions of single data values.
As allowing for navigation wrt.\ tuples of data values renders the logic undecidable, 
we introduce ND-LTL, an extension of BD-LTL by a restricted form 
of tuple-navigation.
While complete ND-LTL is still undecidable, the two natural fragments allowing for either future or past navigation along data values are shown to be Ackermann-hard, yet decidability is obtained by reduction to nested multi-counter systems.
To this end, we introduce and study nested variants of data automata as an intermediate model simplifying the constructions.
To complement these results we show that imposing the same restrictions on BD-LTL yields two 2ExpSpace-complete fragments while satisfiability for the full logic is known to be as hard as reachability in Petri nets.
\end{abstract}

\section{Introduction}
Executions of object-oriented and concurrent systems can naturally be modeled using data words.
They are composed of labels from a finite alphabet together with a data value from an infinite domain.
They can, for example, be considered as an interleaving of actions of an unbounded number of objects or processes, distinguished by identifiers.
Recently, several formalisms based on first-order logic \cite{DBLP:journals/tocl/BojanczykDMSS11,DBLP:journals/corr/abs-1110-1439}
or temporal logic \cite{DBLP:journals/iandc/DemriLN07,DBLP:journals/tocl/DemriL09,DBLP:conf/lics/DemriFP13} have been proposed to specify properties over data words.
Automata-based models have also been considered \cite{DBLP:journals/tocl/NevenSV04,DBLP:journals/tcs/KaminskiF94,DBLP:journals/iandc/BouyerPT03,DBLP:journals/tcs/BjorklundS10} including data automata (DA) \cite{DBLP:journals/tocl/BojanczykDMSS11}.
Usually, in these formalisms the data values can only be
compared with respect to equality. 
More expressive relations like ordering
lead fast to undecidability.
The automata/logic connection has been studied extensively.
For example, the satisfiability problem of 
two-variable first-order logic over data words was shown decidable
by a reduction to the emptiness problem of DA \cite{DBLP:journals/tocl/BojanczykDMSS11}.
They consist of a finite-state letter-to-letter
transducer $\mathcal{A}$ and a class automaton $\mathcal{B}$.
$\mathcal{A}$ changes the labels from the finite alphabet
of the input data word before the data word is projected into 
class strings (one for each different data value) which must all be accepted by $\mathcal{B}$.
Emptiness of \DA was proven decidable by a reduction to 
the reachability problem in multi-counter systems (Petri nets, VASS)
showing a deep connection between data word formalisms and
counter systems (see also \cite{DBLP:conf/lics/DemriFP13,DBLP:conf/fossacs/TzevelekosG13,DBLP:conf/mfcs/BjorklundB07}).

We study \emph{multi-attributed data words}
where, instead of one data value, \emph{several} data
values are associated to a given position. 
This important extension allows
for example modeling \emph{nested} 
parameterized systems where a process
has subprocesses which have subprocesses and so on.
We built on the logic on multi-attributed data words
\emph{basic data LTL (BD-LTL)}
\cite{DBLP:conf/fsttcs/KaraSZ10} allowing for navigation wrt.\ a data value.
It uses the well-known 
LTL with past-time operators 
and has additionally a class quantifier over \emph{one data value}
used to bind a current data value and restrict the evaluation of the formula
to the positions where the same data value appears.
Decidability of the satisfiability problem was shown using a 
reduction to non-emptiness of DA. 
Adding a class quantifier over tuples makes BD-LTL undecidable like 
other logics over multi-attributed data words with tuple navigation
\cite{DBLP:conf/lics/DemriFP13,DBLP:conf/mfcs/BjorklundB07}.

\paragraph{Contributions.}
\begin{wrapfigure}{r}{0.52\textwidth}
\vspace*{-3ex}
\includegraphics[width=0.5\textwidth]{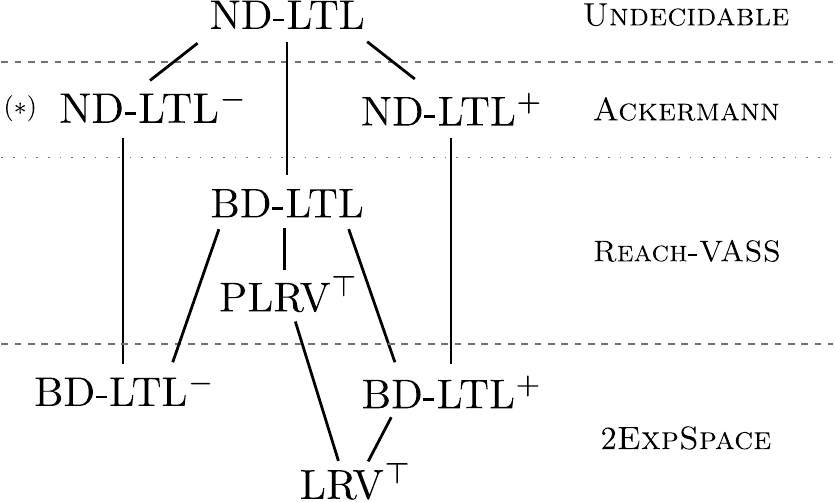}
  \caption{\small Overview of the logics studied in this paper. Lines are drawn downwards to logics with lower expressiveness. The depicted complexity classes apply over finite as well as infinite words except for $\NDLTL^-$, marked by $(*)$, which is undecidable over infinite words.
\label{fig:overview}}
\vspace*{-3ex}
\end{wrapfigure}
We consider first two fragments of \BDLTL:
the \emph{class future fragment} \fBDLTL (past operators are disallowed
for navigation wrt.\ a data value) and the 
\emph{class past fragment} \pBDLTL (restriction of future operators).
Both fragments are shown \textsc{2ExpSpace}-complete
using \cite{DBLP:conf/lics/DemriFP13}
and revisiting the translation from \BDLTL to
DA \cite{DBLP:conf/fsttcs/KaraSZ10}.
Instead of going to general \DA we translate
\fBDLTL and \pBDLTL into \pDA and \sDA, respectively, whose
emptiness problems are in \textsc{ExpSpace}.
In \pDA (resp. \sDA) the language of the class automaton is suffix- (resp.
prefix-) closed allowing
to use the \textsc{ExpSpace}-complete coverability problem of 
multi-counter systems instead of its reachability problem
for which no primitive recursive algorithm is known 
(cf.\ \cite{DBLP:conf/popl/Leroux11}).
We consider both finite and infinite word semantics of the fragments.

We then define the new logic \NDLTL allowing
for navigation wrt.\ tuples respecting a certain tree-order,
i.\,e., there are several layers of data with
nested access. 
For example, one can navigate on the first layer and, fixing a value, navigate on the second (see example below).
Independent navigation on the whole second layer is not possible.
While even with this restricted navigation \NDLTL is undecidable we obtain, as for \BDLTL, two natural fragments $\NDLTL^+$ and $\NDLTL^-$.
We can proof their decidability by a translation into 
\emph{nested data automata (NDA)} that we introduce as an appropriate extension of \DA.
$k$-NDA have $k$ class automata and 
accept data words with $k$ data values at each position. 
The $i$-th class automaton must accept all class
strings obtained by projection of the data word using the same \emph{first} 
$i$ data values.
Emptiness of $k$-NDA is undecidable, but shown decidable
for $k$-sNDA (where class automata have suffix-closed languages) 
and $k$-pNDA (prefix-closed) using 
\emph{nested multi-counter systems} (similar to
models in \cite{DBLP:conf/mfcs/BjorklundB07,DBLP:conf/ershov/LomazovaS99})
which generalize multi-counter systems to several layers 
of nested counters. Their emptiness problem is
undecidable, but, as they are well-structured transition
systems \cite{DBLP:journals/tcs/FinkelS01,DBLP:journals/bsl/Abdulla10}, 
coverability and control state reachability are decidable.
$\NDLTL^+$ and $\NDLTL^-$ are shown \textsc{Ackermann}-hard
via a reduction from the control state reachability problem 
of reset multi-counter systems \cite{DBLP:conf/mfcs/Schnoebelen10}.
Finally, $\NDLTL^+$ is decidable over infinite words
but $\NDLTL^-$ is not.
Figure \ref{fig:overview} summarizes some results.

\paragraph{Related Work.}
The logics LRV$^\top$ (based on \cite{DBLP:journals/logcom/DemriDG12})
and the more expressive LRV over
multi-attributed data words studied in \cite{DBLP:conf/lics/DemriFP13}
built also on LTL and allow to state that \emph{one} 
of the current data values must be seen again
in the future. %
LRV (LRV$^\top$) can be extended to PLRV (PLRV$^\top$)
with past obligations. 
PLRV$^\top$ is less expressive than \BDLTL\footnote{In \cite{DBLP:conf/lics/DemriFP13} it is stated without proof that PLRV is also less expressive
than \BDLTL.}
and we show that LRV$^\top$ is less
expressive than \fBDLTL. LRV (and LRV$^\top$) are 
\textsc{2ExpSpace}-complete like \fBDLTL.
We use their hardness result for our logic.
The proof of the upper bound
is also based on the coverability problem of multi-counter systems. 
However, our proof is split into smaller, structured parts.
The handling of infinite word versions of our fragments is similar to 
theirs but we have to treat the additional problems coming from
the nested data.
Navigation wrt.\ data tuples was considered and shown undecidable
but no decidable fragments were given.
A logic handling data values in a very natural way is Freeze-LTL \cite{DBLP:journals/tocl/DemriL09}.
It exhibits a similar future-restriction as \fBDLTL and $\NDLTL^+$ and finite satisfiability is decidable and \textsc{Ackermann}-hard.
However, satisfiability over infinite words is undecidable while it is still decidable for \fBDLTL and $\NDLTL^+$.
In \cite{DBLP:conf/mfcs/BjorklundB07}, words with nested data values 
were also considered. They show undecidability for the two-variable logic
with two layers of nested data
and the $+1$ and $<$ predicates over positions.
They introduce \emph{higher-order multi-counter automata}, a very similar model
to our nested multi-counter systems.
Their proof of Turing completeness could be
easily adapted to nested multi-counter systems.
However, the well-structuredness of the model is not exploited.
If the $+1$ predicate is dropped they obtain decidability,
which is orthogonal to our result as we can express the
successor relation in our fragments.
In \cite{DBLP:conf/fossacs/TzevelekosG13} history register automata (HRA) 
have been introduced, which can easily be simulated by our $\pNDA$.
A weak variant of HRA is defined which is similar to our $\pDA$, but only studied over finite words.

\paragraph{Example.}\label{sec:examples}
In object-oriented programming languages, iterating over a list is usually done using a method \emph{next} on a corresponding \emph{iterator object}.
Once the state of a list changes, e.\,g., by adding an element, any iterator for that list created before should no longer be used.
We model this scenario using propositions $\mathit{newItr}$, $\mathit{add}$, $\mathit{next}$ and data words with ordered attributes  $l < s < i$ for identifying the list, the list's state and the iterator, respectively.
Thus, fixing a state, fixes also the list it belongs to and fixing an iterator object fixes the corresponding list in its current state.

Consider two constraints: (1) When observing $\mathit{add}$, the state of the list changes, i.\,e., we observe a \emph{fresh} state ID. (2) When calling $\mathit{next}$, the state ID must not have changed since the creation of the currently used iterator.
By $\G (\mathit{add} \rightarrow \C_s \neg \Y^= \top)$ we can express (1).
We bind the current state and list IDs using a class quantifier $\C_s$ and check that there is no previous position with the same IDs.
We express (2) by $\G (\mathit{next} \rightarrow \C_i (\top \S^= \mathit{newItr}))$.
$\C_i$ binds the current state, list and iterator ID. 
The formula $(\top \S^= \mathit{newItr})$ guarantees that there is a previous position with the same IDs where the iterator is created.
For both constraints we use $\NDLTL^-$ with local past operators ($\Y^=$ and $\S^=$).

\section{Preliminaries}

Let $\mathbb{N}=\lbrace0,1,2,…\rbrace$ be the set of natural numbers and $[k]:=\lbrace1,…,k\rbrace$ for $k\in\mathbb{N},k>0$.
We denote the set of finite words over an alphabet $\Sigma$ by $\Sigma^*$, the set of infinite words by $\Sigma^\omega$ and their union by $\Sigma^\infty = \Sigma^* \cup \Sigma^\omega$.
The empty word is denoted $\epsilon$.
The \emph{shuffle} of two words $w,w' \in \Sigma^\infty$ is inductively defined by $\epsilon \shuffle w = w \shuffle \epsilon = \{w\}$ and $a w \shuffle a' w' = a(w \shuffle a'w') \cup a'(aw \shuffle w')$ where $a,a'\in \Sigma$.
The shuffle of two languages $L,L'\subseteq\Sigma^\infty$ is $L \shuffle L' = \{w \shuffle w' \mid w \in L, w' \in L'\}$ and $\shuffle(L) = \bigcup \{M \mid M \subseteq L \shuffle M\}$ denotes the infinite shuffle of a language with itself.
For two sets $M, N$ we denote by $M^N$ the set of all mappings $f:N \to M$ from $N$ to $M$.
Given a partial order $(M,\sqsubseteq)$ we write $m^\downarrow = \lbrace m'\in M\mid m'\sqsubseteq m\rbrace$ for the downward closure of $m\in M$.
We define a \emph{tree order} $(M,\le)$ to be a partial order s.\,t.\ for all $m\in M$ its downward closure is a linear order $(m^\downarrow, \le)$.
Hence, we allow a tree order to contain several minimal elements (roots).

An $\infty$-automaton over a finite \emph{input alphabet} $\Sigma$ is a tuple $\mathcal{A} = (Q, \Sigma, \delta, I, F, B)$ 
where $Q$ is a finite set of \emph{states}, $I,F,B \subseteq Q$ are sets of \emph{initial}, \emph{final} and \emph{Büchi-accepting} states, respectively, and $\delta \subseteq Q \times \Sigma \times Q$ is the \emph{transition relation}. 
A \emph{run} of $\mathcal{A}$ on a word $w_0w_1…\in\Sigma^\infty$, $w_i\in\Sigma$ is a maximal sequence of transitions $t_0t_1…\in\delta^\infty$ with $t_i=(q_i,w_i,q_{i+1})$ and $q_0\in I$.
It is \emph{accepting} if it ends in a final state $q_f\in F$ or visits a Büchi-accepting state $q_b\in B$ infinitely often.
$\mathcal{A}$ \emph{accepts} $w$ if there is an accepting run of $\mathcal{A}$ on $w$ and the set of all accepted words is denoted $\mathcal{L}(\mathcal{A})$.

A \emph{letter-to-letter transducer} is an $\infty$-automaton $\mathcal{T}=(Q,\Sigma,\Gamma,\delta,I,F,B)$ where $\Gamma$ is an additional \emph{output alphabet} and $\delta\subseteq Q\times\Sigma\times\Gamma\times Q$ is a \emph{transition relation with output}.
A word $\gamma\in\Gamma^\infty$ is an \emph{output} of $\mathcal{T}$ if there is an accepting run of $\mathcal{T}$ labeled by $\gamma$.
For $w\in\Sigma^\infty$ we denote $\mathcal{T}(w)\subseteq\Gamma^\infty$ the set of possible outputs of $\mathcal{T}$ when reading $w$.

\subsubsection{Data Words and Data Languages.}
Let $\Sigma$ be a finite alphabet, $\Delta$ an infinite set of \emph{data values} and $A$ a finite set of \emph{attributes}. 
A \emph{multi-attributed data word} is a finite or infinite sequence $w=w_0w_1…\in(\Sigma\times\Delta^A)^\infty$ of pairs $w_i = (a_i, \vec{d}_i)$ of letters and \emph{data valuations} $\vec{d}_i:A \to \Delta$.
Given a valuation $\vec{d} \in \Delta^A$ and a set of attributes $X \subseteq A$ we denote by $\vec{d}|_X$ the restriction of $\vec{d}$ to $X$.
We call $\str{w} := a_0a_1\ldots \in \Sigma^\infty$ the \emph{string projection} of $w$. 
The \emph{$X$-class string} of $w$ for a data valuation $\vec{d} \in \Delta^X$ is the maximal projected subsequence $\class{w}{\vec{d}} := a_{i_0}a_{i_1}… \in \Sigma^\infty$ of $w$ with $0 \le i_j \le |w|$, $i_j<i_{j+1}$ and $\vec{d}_{i_j}|_X=\vec{d}$.
We use natural numbers $1,2,3,…$ as representatives for arbitrary data values.
For a data word $w = (a_0,\vec{d}_0)(a_1,\vec{d}_1)…$ we also write $\binom{a_0a_1…}{\vec{d}_0\vec{d}_1…}$.
For $|A|=1$ we call data words $w\in(\Sigma\times\Delta^A)^\infty$ \emph{single-attributed}.
We may then omit the functional notation and use $\Delta$ instead of $\Delta^A$ if $A$ is not essential, e.\,g., writing $w\in(\Sigma\times\Delta)^\infty$.

\paragraph{Register Automata (RA).}
A \emph{register automaton} \cite{DBLP:journals/tcs/KaminskiF94} over $\Sigma$ and $\Delta$ is a tuple $\mathcal{R} = (Q, \Sigma, k, \delta, I, F, B)$ where $Q$ is a finite set of states, $I,F,B \subseteq Q$ are sets of initial, final and Büchi-accepting states, respectively, $k \geq 1$ 
is the number of \emph{registers} and $\delta \subseteq Q \times 2^{[k]} \times 2^{[k]} \times \Sigma \times [k] \times Q$ is the transition relation.
A \emph{configuration} of $\mathcal{R}$ is a pair $(q, v)$ where $q \in Q$ and $v: [k] \to \Delta \cup \{\bot\}$ is a valuation of the registers.
A \emph{run} of $\mathcal{R}$ on a single-attributed data word $w=(a_0,d_0)(a_1,d_1)…\in(\Sigma\times\Delta)^\infty$ is a maximal sequence of configurations $\rho = (q_0,v_0)(q_1,v_1)…$ s.\,t.\ $q_0\in I$ and for all $0\le i<|w|$ there is a transition $(q_i,R^=_i,R^{\neq}_i,a_i,x_i,q_{i+1}) \in \delta$ such that $\forall_{r\in R^=_i} v_i(r) = d_i$, $\forall_{r\in R^{\neq}_i} v_i(r) \neq d_i$, $v_{i+1}(x_i) = d_i$ and $\forall_{r \neq x_i} v_{i+1}(r) = v_i(r)$.
A run $\rho$ of $\mathcal{R}$ is \emph{accepting} if it ends in a final state $q\in F$ or it visits a Büchi-accepting state $q\in B$ infinitely often.
An \RA accepts a single-attributed data word $w$ if it has an accepting run on $w$.

\subsubsection{Multi-counter Systems.}
A \emph{reset multi-counter system (rMCS)} is a tuple $\mathcal{M} = (Q, C, \delta, Q_0)$ where $Q$ and $C$ are finite sets of \emph{(control) states} and \emph{counters}, respectively, $Q_0\subseteq Q$ is the set of \emph{initial states} and, for $OP:=\lbrace\mathrm{inc}, \mathrm{dec}, \mathrm{res}\rbrace$, $\delta\subseteq Q \times OP \times C \times Q$ is the \emph{transition relation}.
A \emph{run} of $\mathcal{M}$ is a sequence $\rho \in Q_0 \times (OP \times C \times Q)^\infty$, s.\,t.\ every subsequence $(q,op,c,q')$ of $\rho$, with $q,q'\in Q$, $op\in OP$, $c\in C$, is an element of $\delta$ and counters never become negative, i.\,e, there is an injection $f_\rho: \mathbb{N} \to \mathbb{N}$ that maps every position $i$ in $\rho$ with $(\rho_i,\rho_{i+1}) = (\mathrm{dec},c)$, for $c\in C$, to a position $j<i$ with $(\rho_j,\rho_{j+1}) = (\mathrm{inc},c)$ and $(\rho_k,\rho_{k+1})\neq(\mathrm{res},c)$, for all $k$ with $j<k<i$.
An \emph{MCS} is an rMCS where the transition relation does not use the reset operation $\mathrm{res}$.

\section{Local Navigation in \BDLTL}
\label{sec:bdltl-fragments}

The temporal logic \BDLTL is based on \LTL. Linear-time properties are formulated using temporal operators to navigate along the positions of a word.
This concept is extended analogously to data words by allowing for navigation along the occurrences of a data value.
While the \LTL operators express properties on the global structure of the word, independent of associated data values, navigation along the class strings of a word allows for expressing a local view, e.\,g., modeling the behaviour of a single process.

We now recall syntax and semantics of \BDLTL \cite{DBLP:conf/fsttcs/KaraSZ10} and define two natural fragments $\BDLTL^+$ and $\BDLTL^-$ where local navigation is restricted to future and past operators, respectively.
The satisfiability problem of \BDLTL is decidable. 
Yet, it is known to be as hard as reachability in Petri nets \cite{DBLP:conf/fsttcs/KaraSZ10} and we show that satisfiability in our fragments is still \textsc{2ExpSpace}-hard.
The next section then sharpens this result by developing a \textsc{2ExpSpace} decision procedure based on restricted variants of data automata.

Let $AP$ be a finite set of atomic propositions and $A$ a finite set of attributes.
The syntax of \BDLTL formulae consists of \emph{position formulae} $\varphi$ and \emph{class formulae} $\psi$.
It is defined by the following grammar where $p\in AP$, $x,y \in A$ and $r\in\mathbb{Z}$.
\begin{align*}
  \varphi &::= p  \mid \varphi \land \varphi \mid\neg\varphi\mid \X \varphi   \mid\Y \varphi   \mid \varphi \U \varphi   \mid\varphi\S \varphi   \mid\C_x^r\psi\\
  \psi &::= @x \mid \psi \land \psi \mid\neg\psi\mid \X^= \psi \mid\Y^= \psi \mid \psi \U^= \psi \mid\psi\S^= \psi \mid\varphi 
\end{align*}

The semantics of \BDLTL position formulae $\varphi$ is defined over models $(w,i)$ consisting of an $A$-attributed data word $w=(a_0,\vec{d}_0)(a_1,\vec{d}_1)…\in(\Sigma\times\Delta^A)^\infty$ over alphabet $\Sigma=2^{AP}$ and and a position $0\le i<|w|$.
Class formulae $\psi$ are defined over models $(w,i,d)$ containing an additional data value $d\in\Delta$.
Boolean and \LTL operators are defined as usual, ignoring the data values.
For the semantics of the quantifier $\C^r_x$ and class formulae $\psi$, let $\mathrm{pos}_d(w) := \lbrace i\mid0\le i<|w|, \exists_{x\in A}:\vec{d}_i(x) = d\rbrace$ denote the set of positions $i$ in $w$ where some attribute has the value $d\in\Delta$.
Then,
\[
\begin{array}{lll}
   (w,i) &\models \C^r_x \psi &\text{if } 0\le i+r<|w| \text{ and }(w,i+r,\vec{d}_i(x))\models\psi,\\
  (w,i,d)&\models \varphi        &\text{if } (w,i) \models \varphi,\\
  (w,i,d)&\models @x       &\text{if } \vec{d}_i(x) = d,\\
  (w,i,d)&\models \X^=\psi    &\text{if there is } j\in\mathrm{pos}_d(w), j>i\\ 
                    &&\text{ and, for the smallest such } j,\ (w,j,d) \models \psi,\\
  (w,i,d)&\models\psi_1\U^=\psi_2&\text{if there is } j\in\mathrm{pos}_d(w), j\ge i 
                      \text{ s.\,t. } (w,j,d)\models\psi_2 \\
                    &&\text{and } \forall_{j'\in\mathrm{pos}_d(w), j>j'\ge i}: (w,j',d)\models\psi_1.
\end{array}
\]

The operators $\Y^=$ and $\S^=$ are furthermore defined as expected and $(w,0)\models\varphi$ is abbreviated $w\models\varphi$.
We also use the abbreviations $\top$ and $\F^=\varphi := \top\U^=\varphi$.

\begin{definition}[$\BDLTL^\pm$]
We define the following syntactical fragments:
\BDLTL without operators $\X^=$ and $\U^=$ is called $\BDLTL^-$.
\BDLTL without operators $\Y^=$ and $\S^=$ is called $\BDLTL^+$.
\end{definition}

In \cite{DBLP:conf/lics/DemriFP13}, the \emph{Logic of Repeating Values (\LRV)} was introduced as an extension of \LTL interpreted over multi-attributed data words. 
The additional operators are of the form $x\approx\X^ry$, $x \approx \langle\varphi?\rangle y$ and $x \not\approx \langle\varphi?\rangle y$.
The former expresses that the current value of attribute $x$ must be equal to the value of attribute $y$ at the position $r$ steps ahead.
Similarly, the latter two express that the value of $x$ must eventually or never, respectively, be observed as the value of $y$ at a position where, in addition, a formula $\varphi$ holds.
In LRV$^\top$ only $x\approx\X^ry$ and $x \approx \langle\top?\rangle y$ are allowed. 
$x\approx\X^ry$ and
$x \approx \langle\varphi?\rangle y$
can easily be translated into $\BDLTL^+$: 
$x \approx \X^ry$ is equivalent to $\C_x^r@y$ and 
$x \approx \langle\varphi?\rangle y$ is equivalent to $\C_x^0\X^=\F^=(@y \land \varphi)$ \cite{DBLP:conf/fsttcs/KaraSZ10}.
On the contrary, \LRV cannot express the operator $\X^=$.

\begin{proposition}
  \fBDLTL is strictly more expressive than LRV$^\top$.
\end{proposition}
The satisfiability problem of $\LRV^\top$ (and \LRV) was shown to be \textsc{2ExpSpace}-hard in~\cite{DBLP:conf/lics/DemriFP13} by encoding runs of so called chain automata using exponentially many counters.
The proof \cite[Lemma 15]{DBLP:conf/lics/DemriFP13} can easily be adapted to show that the variant of $\LRV^\top$ where past \emph{instead} of future operators are used ($x\approx\Y^ry$, $x \approx \langle\varphi?\rangle^{-1}y$) is also \textsc{2ExpSpace}-hard and as $\BDLTL^-$ subsumes this variant we obtain a lower bound for both of our fragments.

\begin{theorem}[Hardness]\label{thm:bdltl-hardness}
  The satisfiability problems of \fBDLTL and \pBDLTL are \textsc{2ExpSpace}-hard over both, finite and infinite data words.
\end{theorem}

\section{Satisfiability of $\BDLTL^\pm$ is \textnormal{\textsc{2ExpSpace}}-complete}
\label{sec:bdltl-sat}

This section is dedicated to an exact characterization of $\BDLTL^\pm$ satisfiability in terms of complexity.
It also provides a basis for Section~\ref{sec:ndltl-sat} that follows a similar structure but is technically more involved.

First, we formally define data automata and give restrictions that reflect the restrictions on our logic.
They allow us to decide emptiness in \textsc{ExpSpace}, as opposed to full data automata for which emptiness is as hard as reachability in Petri nets~\cite{DBLP:journals/tocl/BojanczykDMSS11}.
Second, we briefly recall the (exponential) translation from \BDLTL to data automata~\cite{DBLP:journals/corr/abs-1010-1139} and show that our logical restrictions indeed carry over to the restrictions on the automata side.

\subsection{\textnormal{\textsc{ExpSpace}}-variants of Data Automata}

A \emph{data automaton (\DA)} is a tuple $\mathcal{D} = (\mathcal{A},\mathcal{B})$ where
the \emph{base automaton} $\mathcal{A} = (Q,\Sigma,\Gamma,\delta_\mathcal{A},Q_0,F_\mathcal{A},B_\mathcal{A})$ is a letter-to-letter transducer and the \emph{class automaton} $\mathcal{B} = (S,\Gamma,\delta_\mathcal{B}, I, F, B)$ is an $\infty$-automaton.
A \emph{memory function} of $\mathcal{D}$ is a mapping $f:\Delta \to S \cup \lbrace\bot\rbrace$ and we denote $\mathfrak{F}$ the set of all memory functions.
A \emph{configuration} of $\mathcal{D}$ is a tuple $(q,f)\in Q\times\mathfrak{F}$ consisting of a base automaton state and a memory function.
A \emph{run} of $\mathcal{D}$ on a single-attributed data word $w=(a_0,d_0)(a_1,d_1)…\in(\Sigma\times\Delta)^\infty$ is a maximal sequence $\rho=(q_0,f_0)(q_1,f_1)…\in(Q\times\mathfrak{F})^\infty$ such that $q_0\in Q_0$, $\forall_{d\in\Delta}:f_0(d)=\bot$ and for all consecutive positions $i,i+1$ on $\rho$ there is a transition $(q_i,a_i,g,q_{i+1})\in\delta_\mathcal{A}$ of the base automaton and a transition $(s,g,s')\in\delta_\mathcal{B}$ of the class automaton such that 
\begin{inparaenum}[(1)]
  \item $f_{i+1}(d_i) = s'$ and 
  \item either $f_i(d_i)= s$, or $f_i(d_i) = \bot$ and $s\in I$, and 
  \item $\forall_{d'\in\Delta,d'\neq d_i}: f_i(d') = f_{i+1}(d')$.
\end{inparaenum}

The run $\rho$ is \emph{accepting} if 
\begin{inparaenum}[(I)]
  \item it ends in a configuration $(q,f)$ with $q\in F_\mathcal{A}$ is final and $f(\Delta)\cap S\subseteq F$, or 
  \item there are infinitely many configurations $(q,f_i)$ on $\rho$ such that $q\in B_\mathcal{A}$ is Büchi-accepting and for each data value $d$ occurring \emph{last} at some  position $i$ on $w$ the state $f_{i+1}(d)\in F$ is final and for each data value $d'$ occurring \emph{infinitely} often on $w$ there are infinitely many positions $j$ with $d_j=d'$ and $f_{j+1}(d')\in B$ is Büchi-accepting.
\end{inparaenum}
The word $w$ is accepted if there is an accepting run  of $\mathcal{D}$.

Intuitively, the base transducer $\mathcal{A}$ reads a letter $a_i\in\Sigma$, performs a transition and outputs its label $g\in\Gamma$.
The memory function maintains an \emph{instance} of the class automaton $\mathcal{B}$ for every data value that occurred so far and spawns a new instance for a fresh data value.
The (present or newly spawned) instance of $\mathcal{B}$ that corresponds to the current data value $d_i$, reads $g$ and performs a step.
For $\mathcal{D}$ to accept, $\mathcal{A}$ and every spawned instance of $\mathcal{B}$ needs to accept by either terminating in a final state or visiting some Büchi-accepting state infinitely often.

\begin{definition}[Prefix- and suffix-closed DA]
  A data automaton $\mathcal{D}=(\mathcal{A},\mathcal{B})$ is \emph{locally prefix-closed (\pDA)} if all states of the class automaton $\mathcal{B}$ are final and Büchi-accepting.
 It is \emph{locally suffix-closed (\sDA)} if all states of $\mathcal{B}$ are initial.
\end{definition}


The construction to decide emptiness of \DA given in \cite{DBLP:journals/tocl/BojanczykDMSS11} translates a \DA into a multi-counter automaton (MCA) that maintains for every class automaton state the number of instances residing in it.
That way, emptiness of \DA reduces (for finite words) to reachability in MCA.
Note that technical differences in the various notions of counter systems (e.\,g., MCA, MCS, VASS, Petri nets) are inessential here.

For \pDA, where all class automaton states are final and Büchi-accepting, automaton instances can be dismissed in any state.
The corresponding MCS thus allows for a random decrement of counters.
Clearly, in such a \emph{lossy} system the problem of reachability reduces to \emph{coverability}.
Regarding infinite words, \emph{repeated coverability} is sufficient since every class automaton state is also Büchi-accepting. 
Both problems are in \textsc{ExpSpace}~\cite{DBLP:journals/tcs/Rackoff78,DBLP:conf/apn/Habermehl97}.

For an \sDA we can decide if it accepts a \emph{finite} word by reversing the automata and
checking the resulting \pDA for emptiness.
In the rest of this section we address the remaining case of \sDA emptiness wrt.\ \emph{infinite} words obtaining the following result.

\begin{theorem}\label{thm:expspace-da}
  Emptiness of \pDA and \sDA over finite and infinite data words is decidable in \textsc{ExpSpace}.
\end{theorem}

Let $\mathcal{D} = (\mathcal{A},\mathcal{B})$ be an \sDA with $\mathcal{A}=(Q,\Sigma,\Gamma,\delta_\mathcal{A},Q_0,\emptyset,B_\mathcal{A})$ (we omit final states) and $\mathcal{B}=(S,\Sigma,\delta,I,F,B)$.
Towards deciding emptiness of $\mathcal{D}$, we consider an accepting run $\rho$ of $\mathcal{D}$ and separate the finite from the infinite behaviour in terms of transitions $t\in\delta$ of the class automaton:
There is a position $i$ on $\rho$, such that $t$ is taken after $i$ iff $t$ is taken infinitely often on $\rho$.

The idea is now to guess (characteristics of) the configuration at this position and check that there is a finite run reaching the configuration and that starting from it there is an infinite accepting run.
For the former, we construct an \sDA that accepts a \emph{finite} word iff the configuration is reachable.
For the latter, we now have guaranteed infinite recurrence of all relevant transitions and can thereby reduce the problem to emptiness of an at most exponentially larger Büchi automaton.

For a set $T\subseteq\delta$ of transitions of the class automaton $\mathcal{B}$ and a state $q\in Q$ of the base automaton $\mathcal{A}$, consider the following three properties.
\begin{enumerate}[({A}1)]
  \item After taking any transition in $T$, $\mathcal{B}$ can eventually reach a
             final state from $F$ or an accepting state from $B$, only by taking
             transitions from $T$.
  \item There is a sequence $t_1t_2…\in T^\omega$ with 
             $t_i=(s_i,g_i,s_i')$ in which each $t\in T$ occurs infinitely 
             often and $g_1g_2…\in\Gamma^\omega$ is an output of $\mathcal{A}$ starting in $q$.
  \item There is a reachable configuration $(q,f)$ such that for all data
             values $d\in\Delta$, either
  \begin{inparaenum}[(i)]
    \item there is no corresponding instance of $\mathcal{B}$ ($f(d) = \bot$), or
    \item the corresponding instance of $\mathcal{B}$ is in an accepting state 
          ($f(d)\in F$), or
    \item there is a transition $(f(d),g,s)\in T$ for some $s\in S$ and some $g\in\Gamma$.
  \end{inparaenum}
\end{enumerate}

\begin{lemma}\label{lem:sda-properties}
  The \sDA $\mathcal{D}$ accepts an infinite data word iff there are $T\subseteq\delta$, $q\in Q$ such that the properties (A1)--(A3) hold.
\end{lemma}

($\Rightarrow$) For an accepting run of $\mathcal{D}$ take $T$ to be the set of transitions of $\mathcal{B}$ taken infinitely often on $\rho$.
Let $i$ be a position after which only transitions from T are taken and $q_i\in Q$ the base automaton state at position $i$.
If (A1) did not hold because of some transition $t\in T$ then the instance of $\mathcal{B}$ performing it after position $i$ would reject.
The suffix of $\rho$ starting from $i$ is a witness for (A2).
The configuration $(q_i,f_i)\in Q\times\mathfrak{F}$ at position $i$ is a witness for (A3) as in particular all instances of $\mathcal{B}$ that terminate before $i$ accept.

($\Leftarrow$)
Given $T$ and $q$ we can construct an accepting run of $\mathcal{D}$.
Property (A3) allows us to find a run to some configuration $c\in Q\times\mathfrak{F}$ from which we are able to continue only using transitions from $T$.
$\mathcal{D}$ can then continue performing the sequence of transitions $\tau\in T^\omega$ provided by (A2) since
all states of $\mathcal{B}$ are initial and hence new instances can be spawned in any state when needed.
Now, (A1) and the fact that each transition in $T$ occurs infinitely often on $\tau$ guarantee that for any (non-terminated) instance of $\mathcal{B}$ we can choose a subsequence of $\tau$ that can be performed by $\mathcal{B}$ in order to accept.

\begin{lemma}\label{lem:sda-construction}
  Given $T\subseteq\delta$ and $q\in Q$, it is decidable in \textsc{ExpSpace} if properties (A1)--(A3) are satisfied.
\end{lemma}

Verifying (A1) is a reachability problem in the finite graph of $\mathcal{B}$ restricted to $T$.
Further, we can build a Büchi automaton over $\Gamma$ that is non-empty iff (A2) holds:
In $\mathcal{A}$, take the outputs as inputs and remove all transitions with a label not occurring on any transition in $T$.
For each transition $(s,g,s')\in T$, intersect the automaton with the property $\G\F g$.
The size of the resulting Büchi automaton is at most $c^{({|Q|}^2)}$ for a constant $c$.
Finally, (A3) can be verified by constructing the \sDA $\hat{\mathcal{D}}=(\hat{\mathcal{A}}, \hat{\mathcal{B}})$ with $\hat{\mathcal{A}} = (Q,\Sigma,\Gamma,\delta_\mathcal{A},Q_0,\lbrace q\rbrace,\emptyset)$ and $\hat{B} =(S,\Sigma,\delta,I, F\cup{^\bullet}T)$, where ${^\bullet}T := \lbrace s\in S\mid\exists_{s'\in S,g\in\Gamma}: (s,g,s')\in T\rbrace$, and checking it for emptiness in exponential space as above.

To conclude, using Lemma~\ref{lem:sda-properties} and~\ref{lem:sda-construction} we can check the \sDA $\mathcal{D}$ for emptiness wrt.\ infinite words by nondeterministically guessing a state $q\in Q$ and a set of transitions $T\subseteq\delta$ and verifying (A1)--(A3) in exponential space.

\subsection{From $\BDLTL^\pm$ to Data Automata}
\label{ssc:bdltl2da}
By revisiting the construction given in \cite{DBLP:conf/fsttcs/KaraSZ10} \fBDLTL and \pBDLTL can be translated into at most exponentially larger \sDA and \pDA, respectively.

The first step is to eliminate multiple attributes by translating a formula into a satisfiability-equivalent one over single-attributed data words.
The basic idea is to encode $A$-attributed data words by using segments of length $|A|$, each representing a single position in the original word.
The temporal operators are adjusted by offsets according to the segment length and positioning information within a fragment is encoded by additional propositions.
The construction, given in detail in \cite{DBLP:journals/corr/abs-1010-1139}, only uses additional operators $\C^r_x$ and a $\BDLTL^\pm$ formula hence stays in the respective fragment. 
Further, it is at most polynomially larger than the original one.

Second, the obtained formula $\varphi$ is translated into a data automaton.
The largest (absolute) value used by $\C^{r}_x$ operators in $\varphi$ is denoted $r_{max}$.
The set $AP_\varphi$ of atomic propositions used by $\varphi$ is extended by propositions $p^\psi_j$ and $=_j$ for each $-r_\mathrm{max}\le j\le r_\mathrm{max}$ and subformula $\psi$ of $\varphi$.
$\mathcal{A}_\varphi$ is supposed to check, that $p^\psi_j$ holds at some position $i$ iff $\psi$ holds at position $i + j$.
A proposition $=_j$ is supposed to hold at position $i$ iff position $i+j$ carries the same data value. 
Checking all this and additionally that $p_\varphi$ holds at the very first position, $\mathcal{A}_\varphi$ accepts exactly the models of $\varphi$, up to the additional propositions.

Correct occurrence of propositions $p^\psi_0$ where $\psi$ is a position formula can be checked by the base automaton.
In all other cases $\psi$ can be assumed to be of the form $\C^r_x \psi$ where $\psi$ is not a position formula.
These formulae can be checked using the local automaton.
The context information provided for $r_\mathrm{max}$ positions into the past as well as the future allows the local automaton to determine the effect of the $C^r_x$ operator without any additional temporal navigation.

Propositions of the form $p^\psi_j$ with $j \ne 0$ can be handled by the base automaton.
Note that in contrast to the construction given in \cite{DBLP:conf/fsttcs/KaraSZ10} we require these propositions as a suffix or prefix closed local automaton can not keep track of this information itself.
What remains is to verify the correct annotation by the propositions $=_j$.
This can easily be done using a register automaton $\mathcal{R}$ that maintains the frame of data values and verifies the propositions.
While it is known that \RA can be translated into \DA it is not clear how to adapt the construction given in \cite{DBLP:conf/fsttcs/KaraSZ10} for \sDA and \pDA.

\begin{lemma}[Simulating RA]\label{lem:ra}
Given an \sDA or \pDA $\mathcal{D}$ and a register automaton $\mathcal{R}$, we can construct an \sDA or \pDA $\mathcal{D}'$, respectively, such that $\mathcal{L}(\mathcal{D})\cap\mathcal{L}(\mathcal{R}) = \emptyset$ iff $\mathcal{L}(\mathcal{D}') = \emptyset$.
$\mathcal{D}'$ is of polynomial size in the size of $\mathcal{D}$ and $\mathcal{R}$ and of exponential size in the number of registers of $\mathcal{R}$.
\end{lemma}
The basic idea is to extend the alphabet with a new letter $\$$.
At a position where a data value would have been stored in a register $r$, the class automaton changes its state and thereby marks the current data values as being stored in $r$.
The transducer keeps track of all registers currently containing data values.
Before ``storing'' another data value in an already occupied register $r$, the transducer performs an additional step accepting $\$$ and demanding that one instance (the instance associated with the value currently ``stored'' in $r$) changes its state, such that the data value is no longer marked as being stored in $r$.
Both, a suffix or a prefix closed base automaton suffices to implement this approach.

By Lemma \ref{lem:ra}, we can perform all required checks using an \sDA or \pDA, respectively.
Translating \LTL formulae into word automata results in a state space that is at most exponential in the size of the formula and thus the construction gives an up to exponential overall blowup.
Note that we assume a unary encoding for the offsets $r$ in formulae $\C^r_x$. 

By Theorem~\ref{thm:expspace-da}, the translation proofs our completeness result for $\BDLTL^\pm$.
\begin{theorem}[\textnormal{\textsc{2ExpSpace}}-completeness]\label{thm:bdltl-completeness}
  Satisfiability problems of \fBDLTL and \pBDLTL are 2\textsc{ExpSpace}-complete over finite and infinite data words.
\end{theorem}

\section{Ordered Navigation on Multi-attributed Data Words}

As we have seen, multiple attributes do not actually enrich the models of \BDLTL.
They can be eliminated due to the inability of \BDLTL to reason about their interdependencies.
A natural extension is thus to allow for so-called \emph{tuple navigation}, e.\,g.,
by adding an operator $\C^r_{(x,y)}$ binding a tuple instead of single values.
Class operators such as $\X^=$ and $\S^=$ then navigate along the positions of a multi-attributed data word that carry both values.
Unfortunately, it is well-known that such an extension leads to undecidability.
For example, \LRV is known to be undecidable when being extended by tuple navigation \cite{DBLP:conf/lics/DemriFP13}. 
This implies undecidability of such an extension of $\BDLTL^+$ and by similar arguments $\BDLTL^-$.
\begin{proposition}\label{prp:tuple-nav-undecidable}
  The satisfiability problem of $\BDLTL^\pm$ with tuple navigation is undecidable. 
\end{proposition} 
To overcome the restrictions of \BDLTL while maintaining decidability, at least for reasonable fragments, we define the logic \NDLTL.
\begin{definition}[\NDLTL]
  The logic \emph{Nested Data LTL (\NDLTL)} consists of \BDLTL formulae where the set of attributes $A$ is enriched by a tree order relation $\le\subseteq A\times A$. 
  The fragments \emph{$\NDLTL^+$} and \emph{$\NDLTL^-$} are obtained by the same restrictions as for $\BDLTL^+$ and $\BDLTL^-$, respectively.
\end{definition}
The quantifier $\C_x^r$ in \NDLTL binds not only the value of attribute $x\in A$ but also the values of all smaller attributes.
Class operators, such as $\U^=$, then navigate according to this tuple of values respecting, however, the attribute order in the following sense.

For an attribute $x\in A$, with downward-closure $x^\downarrow$ consisting of attributes $x_1<x_2<…<x_n$, a mapping $\vec{d}\in \Delta^{x^\downarrow}$ induces a vector of data values $ ( \vec{d}(x_1), \vec{d}(x_2), …, \vec{d}(x_n))$.
By $\vec{d} \simeq \vec{d}'$ we denote that $\vec{d}$ and $\vec{d}'$ have the same such vector representation.
Note, this can differ from the element-wise equality of the functions.
Using this we define for a data word $w\in(\Sigma\times\Delta^A)^\infty$ and $\vec{d}\in\Delta^{x^\downarrow}$ the set $\mathrm{pos}_\vec{d}(w)$ of positions $i$ in $w$ where there is an attribute $y\in A$ such that $\vec{d} \simeq \vec{d}_i|_{y^\downarrow}$. 
\NDLTL class formulae are interpreted over models $(w,i,\vec{d})$ where $i\in\mathbb{N}$, $0\le i<|w|$, is a position in $w$ and $\vec{d}\in \Delta^{x^\downarrow}$ for some $x\in A$.
For position formulae $\varphi$, $x,y\in A$ and $r\in\mathbb{Z}$, we define the semantics of the $\C^r_x$ operator and class formulae $\psi$ as follows.
\[
\begin{array}{lll}
  (w,i)         &\models\C^r_x \psi &\text{if }0<i+r<|w| 
                            \text{ and } (w,i+r,\vec{d}_i|_{x^\downarrow})\models\psi,\\
  (w,i,\vec{d}) &\models \varphi       &\text{if } (w,i) \models \varphi,\\
  (w,i,\vec{d}) &\models @x      &\text{if } \vec{d}_i|_{x^\downarrow} \simeq \vec{d},\\
  (w,i,\vec{d}) &\models \X^=\psi   &\text{if there is } j\in\mathrm{pos}_\vec{d}(w),\ j>i,\\
                         & &\text{and, for the smallest such }j,\ 
                            (w,j,\vec{d}) \models \psi,\\
  (w,i,\vec{d}) &\models \psi_1\U^=\psi_2 & \text{if } \exists_{j\in\mathrm{pos}_\vec{d}(w), j\ge i}:
                            (w,j,\vec{d})\models\psi_2 \\
                          &&\text{and }
                            \forall_{j'\in\mathrm{pos}_\vec{d}(w), j>j'\ge i}:
                            (w,j',\vec{d})\models\psi_1.
\end{array}
\]
As before, the operators $\Y^=$ and $\S^=$ are defined as expected. 
The semantics of boolean and \LTL operators in \NDLTL formulae remains as for \BDLTL.

\begin{lemma}\label{lem:encoding-rmcs}
  For every rMCS $\mathcal{M} = (Q,C, \delta, Q_0)$, there is an $\NDLTL^-$ formula $\Phi_\mathcal{M}$ over
  the set of propositions $AP = Q\cup\lbrace\mathrm{inc}, \mathrm{dec}, \mathrm{res}\rbrace\cup C$
  and attributes $A$ s.\,t.\ $\Phi_\mathcal{M}$ is satisfiable iff there is a data word 
  $w\in(2^{AP}\times\Delta^A)^\omega$ where $\str{w} = \lbrace p_0\rbrace\lbrace p_1\rbrace…$ ($p_i\in AP$) and $p_0p_1…$ 
  is a run in $\mathcal{M}$.
\end{lemma}
Using a pair $x_c>\hat{x}_c$ of attributes for each counter $c\in C$, a formula $\bigwedge_{c\in C}\G((\mathrm{res} \land \X c) \to \C^0_{\hat{x}_c}\neg\Y^= \top)$ can be used for specifying resets and $\bigwedge_{c\in C}\G((\mathrm{dec} \land \X \mathrm{c}) \to \C^0_{x_c} \Y^=(\mathrm{inc} \land \X c))$ assures non-negative counter values.
It is clear that using a further constraint of the form $\F q$ allows for expressing control state reachability in rMCS, being \textsc{Ackermann}-hard by results on lossy channel systems in~\cite{DBLP:conf/mfcs/Schnoebelen10}.
Encoding such finite runs of an rMCS \emph{backwards}, can be done analogously within the fragment $\NDLTL^+$.

\begin{theorem}[\textnormal{\textsc{Ack}}-hardness]\label{thm:ack-ndltl}
  Satisfiability of $\NDLTL^\pm$ is \textsc{Ackermann}-hard.
\end{theorem}
Similarly, $\G\F q$ expresses \emph{repeated} control state reachability in rMCS, being undecidable due to results in~\cite{DBLP:conf/stacs/BouajjaniM99}.
Further, full \NDLTL is already undecidable over \emph{finite} words.
This can be shown by considering the formula $\bigwedge_{c\in C}\G((\mathrm{inc} \land \X c) \to \C^0_{x_c}\X^=(\mathrm{dec} \land \X c))$ that ensures that for every incrementing operation, there is a following decrement on the same counter \emph{before the next reset} on that counter.
Thus, reset operations turn into \emph{zero tests}, allowing to encode Minski machine computations where reachability is undecidable.

\begin{theorem}[Undecidability]\label{thm:undecidable-ndltl}
Satisfiability of \NDLTL is undecidable over \emph{finite and infinite} data words.
Satisfiability of $\NDLTL^-$ is undecidable over \emph{infinite} data words.
\end{theorem}

\section{Deciding Satisfiability of $\NDLTL^\pm$}\label{sec:ndltl-sat}

Having established undecidability and hardness results for \NDLTL we finally turn to decision procedures in this section.
We complete our picture by decidability results for the remaining cases of $\NDLTL^-$ over finite words and $\NDLTL^+$ over finite and infinite words.
The structure follows that of Section \ref{sec:bdltl-sat} and we provide the essential ideas for lifting the constructions as well as additional arguments where needed.
To capture the notion of nesting in \NDLTL we extend data automata and again provide restrictions that carry over from the logic.

\subsection{Nested Data Automata}
\label{ssc:nda}

We extend data automata to read multi-attributed data words by adding a class automaton for each attribute.
The class automata are linearly ordered in the sense that the $i$-th class automaton reads refinements (subwords) of the input of the $(i-1)$-th class automaton.
That way they express a linear order on the attributes which is, however, sufficient since we later show that \NDLTL formulae over a tree order can be translated into formulae over a linear order.
For that reason, we only consider attribute sets $[k]=\lbrace1,…,k\rbrace$ for $k\in\mathbb{N}$.

\begin{definition}[Nested data automaton]
A \emph{$k$-nested data automaton} ($k$-\NDA) is a $(k+1)$-tuple $\mathcal{D} = (\mathcal{A},\mathcal{B}_1,…,\mathcal{B}_k)$ where $(\mathcal{A},\mathcal{B}_i)$ is a data automaton for each $i\in[k]$.
$\mathcal{D}$ is called \emph{locally prefix-closed (\pNDA)} if each $(\mathcal{A},\mathcal{B}_i)$ is a \pDA and it is called \emph{locally suffix-closed (\sNDA)} if each $(\mathcal{A},\mathcal{B}_i)$ is an \sDA.
\end{definition}
Let $\mathcal{D} = (\mathcal{A},\mathcal{B}_1,…,\mathcal{B}_k) $ be a $k$-\NDA with $\mathcal{A}=(Q,\Sigma,\Gamma,\delta_\mathcal{A},Q_0, F_\mathcal{A}, B_\mathcal{A})$ and $\mathcal{B}_i = (S_i, \Gamma, \delta_i, I_i, F_i, B_i)$.
A \emph{configuration} of $\mathcal{D}$ is a tuple $c=(q,f_1,…,f_k)\in Q\times\mathfrak{F}_1\times…\times\mathfrak{F}_k$ where $\mathfrak{F}_i$ is the set of \emph{memory functions} $f: \Delta^{[i]} \to S_i\cup\lbrace\bot\rbrace$ (partially) mapping $i$-tuples of data values to states.

A \emph{run} of $\mathcal{D}$ on an $[k]$-attributed data word $w=(a_0,\vec{d}_0)(a_1,\vec{d}_1)…\in(\Sigma\times\Delta^{[k]})^\infty$ is a maximal sequence $\rho=(q_0,f_{1,0}, …, f_{k,0})(q_1,f_{1,1}, …, f_{k,1})…$ of configurations where $q_0\in Q_0$, $f_{i,0}(\Delta^{[i]})=\lbrace\bot\rbrace$ and for each consecutive positions $n,n+1$ on $\rho$ there is a transition $(q_n,a_n,g,q_{n+1})\in\delta_\mathcal{A}$ for $g\in\Gamma$ of the base automaton and a transition $(s_i,g,s_i')\in\delta_i$ for each class automaton $\mathcal{B}_i$ such that
\begin{inparaenum}[(1)]
  \item $f_{i,n+1}(\vec{d}_n|_{[i]}) = s_i'$ and 
  \item either $f_{i,n}(\vec{d}_n|_{[i]})= s_i$, or $f_{i,n}(\vec{d}_n|_{[i]}) = \bot$ and $s_i\in I_i$, and 
  \item $\forall_{\vec{d}'\in\Delta^{[i]},\vec{d}'\neq \vec{d}_n|_{[i]}}: f_{i,n}(\vec{d}') = f_{i,n+1}(\vec{d}')$.
\end{inparaenum}

A run of $\mathcal{D}$ on $w$ is (finitely) \emph{accepting} if it ends in a configuration $(q,f_1,…f_k)$ with $q\in F_\mathcal{A}$ and $\forall_{i\in[k]} f_i(\Delta^{[i]}) \subseteq F_i\cup\lbrace\bot\rbrace$.
Moreover, it is accepting if there are infinitely many configurations $(q,f_{1,n},…,f_{k,n})$ on $\rho$ such that $q\in B_\mathcal{A}$ is Büchi-accepting and for each level $i\in[k]$ and each data valuation $\vec{d}\in\Delta^{[i]}$ there is either 
\begin{inparaenum}[(I)]
  \item \emph{no} position $m$ with $\vec{d}_m|_{[i]} = \vec{d}$, or
  \item a \emph{last} position $m$ with $\vec{d}_m|_{[i]} = \vec{d}$ 
        and the state $f_{i,m+1}(\vec{d})\in F$ is final, or
  \item there are \emph{infinitely many} positions $m$ where 
        $\vec{d}_m|_{[i]} = \vec{d}$ and $f_{i,m+1}(\vec{d})\in B_i$ is 
        Büchi-accepting.
\end{inparaenum}

The idea of deciding emptiness of \pNDA and \sNDA is, again, to translate them into multi-counter systems, which this time will be nested.
Similar notions of such nested systems can be found in~\cite{DBLP:conf/ershov/LomazovaS99,DBLP:conf/mfcs/BjorklundB07}.

\begin{definition}[$k$-nMCS]
A \emph{$k$-nested multi-counter system ($k$-nMCS)} is a tuple $\mathcal{M} = (Q,\delta,I)$ with a finite set of states $Q$, a set of \emph{initial states} $I \subseteq Q$, and a \emph{transition relation} $\delta \subseteq (\bigcup_{i\in[k]} Q^i) \times Q^k$.
\end{definition}

A \emph{multiset} over a set $S$ is a mapping $m\in\mathbb{N}^S$.
For a $k$-nMCS $\mathcal{M} = (Q,\delta,I)$, the set of \emph{configurations of level $i$} are defined inductively (from $k$ to $0$) as $C_k = Q$ and $C_{i-1} = Q \times \mathbb{N}^{C_i}$.
The set of \emph{configurations} of $\mathcal{M}$ is then $C_\mathcal{M} = C_0$.
We can see an element of $C_0$ as a term constructed over unary function symbols $Q$, constants $Q$ and the binary operator $+$.
The terms are considered modulo associativity and commutativity of the $+$ operator which does not appear on the top level.
For example $q_0(q_1(q_3(q_5+q_5+q_6)+q_3(q_6+q_6))+q_1(q_3(q_6+q_6)+q_3(q_6+q_5+q_5))+q_2(q_7(q_8))+q_2(q_7(q_8)))$ corresponds to 
$(q_0, \{(q_1,
    \{(q_3,\{q_5:2,q_6:1\}):2,(q_3,\{q_6:2\}):2\}):1,
(q_2,\{(q_7,\{q_8:1\}):1\}):2\})$.

Now, the transition relation $\to \subseteq C_\mathcal{M} \times C_\mathcal{M}$ on configurations can be easily defined as a rewrite rule. 
For $((q_0,q_1,…,q_i), (q_0',q_1',q_2',…,q_k')) \in \delta$, we have 
$(q_0,X_1+q_1(X_2+… q_i(X_{i+1})…)) \to (q'_0,X_1+q'_1(X_2+… q'_i(X_{i+1}+q'_{i+1}(q'_{i+2}… q'_{k-1}(q'_{k})))))$ where $X_i \in\mathbb{N}^{C_i}$.
As usual we denote by $\to^*$ the reflexive and transitive closure of $\to$.

A \emph{well-quasi-ordering (WQO)} on a set $C$ is a pre-order $\preceq$
such that, for any infinite sequence $c_0,c_1,c_2,…$
there are $i,j$ with $i < j$ and $c_i \preceq c_j$.
A WQO $\preceq$ on a set $C$ induces a WQO $\preceq_m$ on multisets over C as follows.
Let $B = \{b_1,…,b_n\}$ and $B'= \{b'_1,…,b'_{n'}\}$ two multisets over $C$. 
Then, $B \preceq_m B'$ iff there is an injection $h$ from  $[n]$ to $[n']$ with $b_i \preceq b'_{h(i)}$.
Let $\preceq_k$ be the WQO $=$ (equality relation) on the set of states $Q$ of the $k$-nMCS $\mathcal{M}$. 
We iterate the construction and obtain a WQO $\preceq_1$ on $C_\mathcal{M}$.

It can be easily seen that the transition relation $\to$ of $k$-nMCS is \emph{monotonic} wrt.\ $\preceq_1$, i.\,e., if $c_1 \preceq_1 c_2$ and $c_1 \to c_3$ then $c_2 \to c_4$ for some $c_4$ with $c_2 \preceq_1 c_4$.
A $k$-nMCS is hence a \emph{well-structured transition system} \cite{DBLP:journals/bsl/Abdulla10} and we directly obtain the following lemma.

\begin{lemma}[Coverability] \label{lem:coverability}
Let $\mathcal{M} =(Q,\delta,I)$ be a a $k$-nMCS, $c \in C_\mathcal{M}$ a configuration and $q\in Q$ a state.
The \emph{coverability problem} of checking if there is a configuration $c'\in C_\mathcal{M}$ with $c \preceq c'$ such that $(q,\emptyset) \to^* c'$, is decidable.
\end{lemma}

Given a $k$-\NDA $\mathcal{D} = (\mathcal{A},\mathcal{B}_1,…,\mathcal{B}_k)$ where $\mathcal{A} = (Q_0,\Sigma,\Gamma,\delta_0,I_{0},F_{0}, \emptyset)$ and $\mathcal{B}_i = (Q_i,\Gamma,\delta_i,I_{i},F_{i}, \emptyset)$ for $i \in [k]$ with disjoint sets of states and no Büchi-accepting states, we can construct a $k$-nMCS $\mathcal{M}_\mathcal{D} = (\bigcup_{i=0}^k Q_i,\delta,I_0)$ as follows.
Let $((q_0,…,q_i), (q_0',…,q'_k)) \in \delta$ for some $0 \le i \le k$
if there are letters $a\in\Sigma$ and $g\in\Gamma$ such that there is a transition of the base automaton $(q_0,a,g, q'_0) \in \delta_0$ and for all $1 \le j \le i$ we have transitions$(q_j,g,q'_j) \in \delta_j$ of the class automata and for all $j$ with $i < j \le k$ there exists an initial state $q''_j \in I_{j}$ such that $(q''_j,g,q'_j) \in \delta_j$.
Then $\mathcal{D}$ is empty iff a configuration can be reached in $\mathcal{M}_\mathcal{D}$ containing only states from $F := \bigcup_{i=0}^k F_{i}$

In case $\mathcal{D}$ is a \pNDA, all states of the class automata are final and the emptiness problem hence reduces to the coverability problem of $k$-nMCS.
As above, if $\mathcal{D}$ is an \sNDA, we considering the reversal of the base and the class automata to obtain the case of \pNDA (still without Büchi-accepting states).
In the rest of this section we address the remaining case of checking if an \sNDA \emph{with} Büchi-accepting states accepts an infinite data word in order to obtain the following.

\begin{theorem}\label{thm:emptiness-nda}
  Emptiness of \sNDA is decidable over finite and infinite data words.
  Emptiness of \pNDA is decidable over finite data words.
\end{theorem}

Now, let $\mathcal{D} = (\mathcal{A},\mathcal{B}_1,…,\mathcal{B}_k)$ be a $k$-\sNDA where $\mathcal{A} = (Q,\Sigma,\Gamma,\delta_\mathcal{A},I_\mathcal{A},\emptyset, B_\mathcal{A})$ and $\mathcal{B}_i = (S_i,\Gamma,\delta_i,S_i,F_i, B_i)$.
For a configuration $c = (q,f_1,…,f_k)$ of $\mathcal{D}$, a data valuation $\vec{d}\in\Delta^{[1]}$ with $f_1(\vec{d}) \neq \bot$ corresponds to an “active” instance of the class automaton $\mathcal{B}_1$.
Consider the set $m := \lbrace\vec{d}'\in\Delta^{[i]}\mid i\in[k], f_i(\vec{d}') \neq \bot, \vec{d}'(1) = \vec{d}(1)\rbrace$ of data valuations \emph{depending} on $\vec{d}$.
It is prefix-closed wrt.\ the linear order on $[k]$ and can hence be considered as a \emph{tree} with root $\vec{d}$ (level~1).
Define a labeling $s: m \to \bigcup_{i\in[k]}S_i$ attaching to each node $\vec{d}'\in\Delta^{[i]}$ (level~$i$) in $m$ the current state of the corresponding class automaton instance, i.\,e., $s(\vec{d}') := f_i(\vec{d}')$, and repeatedly delete all leaf nodes of $m$ that are final states.
Let $M_c$ be the (finite) set of all such labeled trees $(m,s)$ for a configuration $c$.

As done similar in Section \ref{sec:bdltl-sat}, we characterize a configuration that splits the finite from the infinite behaviour on an accepting run of $\mathcal{D}$.
For a set of transitions $T\subseteq\delta_1$ of $\mathcal{B}_1$, a state $q\in Q$ of $\mathcal{A}$ and a finite set $M$ of finite trees labeled by states from $S_1\cup…\cup S_k$, consider the following properties.
\begin{enumerate}[(B1)]
  \item For all $t_1\in T$ there is a sequence $t_1t_2…\in T^\infty$, $t_i=(s_i,g_i,s_i')$, inducing an accepting run of $\mathcal{B}_1$ and $g_1g_2…\in 
             \shuffle(\mathcal{L}(\mathcal{B}_2)\cap\shuffle(… \cap\shuffle(\mathcal{L}(\mathcal{B}_{k-1})\cap\shuffle\mathcal{L}(\mathcal{B}_k)) … ))$.
  \item There is a sequence $t_1t_2…\in T^\omega$ with $t_i=(s_i,g_i,s_i')$ in which 
        each $t\in T$ occurs infinitely often and $g_1g_2…\in\Gamma^\omega$ is an output of 
        $\mathcal{A}$ starting in $q$.             
  \item There is a reachable configuration $c = (q, f_1,…,f_k)$ with $M = M_c$ 
        such that for all $i\in[k]$ and all $\vec{d}\in\Delta^{[i]}$ either
        \begin{inparaenum}[(i)]
          \item $f_i(\vec{d}) = \bot$ (there is no corresponding instance), or
          \item $\forall_{\vec{d}'\in\Delta^{[k]} \text{ s.\,t. }\vec{d}'|_{[i]}=\vec{d}}
                 \forall_{j\ge i}: f_j(\vec{d}'|_{[j]})\in F_j$ (the corresponding 
                 instance and all instances depending on it are in a final state),
                or
          \item $\exists_{g\in\Gamma,s'\in S}: (f_1(\vec{d}|_{[1]},g,s')\in T$ (there is a transition applicable to
                the corresponding instance of $\mathcal{B}_1$).
        \end{inparaenum}
  \item For each tree $(m,s)\in M$ there is a second labeling $\gamma: m \to \Gamma^\infty$ such that,  
          for the root $r\in m$, the label $\gamma(r)$ is accepted by $\mathcal{B}_1$ restricted 
          to $T$ when starting in state $s(r)\in S_1$ and for all nodes $v\in m$ on a
          level $i>1$
          \begin{inparaenum}[(i)]
            \item $\gamma(v)$ is accepted by $\mathcal{B}_i$ starting in state $s(v)\in S_i$ and
            \item $\gamma(v)$ must be a shuffle of the labels of the direct children 
                  of $v$ and a (possibly infinite) number of words from the 
                  shuffle set 
                  $\shuffle(\mathcal{L}(\mathcal{B}_{i+1})… \cap\shuffle(\mathcal{L}(\mathcal{B}_{k-1})\cap\shuffle\mathcal{L}(\mathcal{B}_k)) … )$.
            \end{inparaenum}
\end{enumerate}
\begin{lemma}\label{lem:snda-properties}
  The \sNDA $\mathcal{D}$ accepts an infinite data word iff there are $T\subseteq\delta_1$, $q\in Q$ and a set $M$ of finite trees labeled by states from $S_1\cup…\cup S_k$ s.t. properties (B1)--(B4) hold.
\end{lemma}

For a complete proof see Appendix \ref{spx:nda}. It is based on similar arguments as Lemma~\ref{lem:sda-properties}.
The new aspect is to schedule class automaton instances on higher levels consistently.

\begin{lemma}\label{lem:construction}
For $T\subseteq\in\delta_1$ and $q\in Q$ we can decide if there is a set $M$ of finite trees labeled by states from $S_1\cup…\cup S_k$ such that the properties (B1)--(B4) hold.
\end{lemma}

Given $T$ we verify (B2) as above by constructing and analyzing a Büchi automaton.
We now sketch the procedure to compute the candidates $M$ that satisfy (B3).
We construct a $k$-\sNDA $\tilde{\mathcal{D}} = (\tilde{\mathcal{A}}, \tilde{\mathcal{B}}_1,…,\tilde{\mathcal{B}}_k)$, without Büchi-accepting states, from $\mathcal{D}$ by taking $q$ as only final state in $\tilde{\mathcal{A}}$.
In each step, $\tilde{\mathcal{A}}$ guesses whether the currently active instance of $\tilde{\mathcal{B}_1}$ performs its last step entering a source state $s$ of some transition $(s,g,s')\in T$.
In that case it marks the current output by some flag.
$\tilde{B}_1$ simulates $\mathcal{B}_1$ and verifies that $\tilde{\mathcal{A}}$ guessed correctly.
Each other class automaton $\tilde{\mathcal{B}}_i$ ($i>1$) simulates $\mathcal{B}_i$. 
Upon reading the flag it moves to an accepting copy of the state they would have moved to otherwise.

The configurations in which $\tilde{\mathcal{D}}$ can accept are exactly those configurations reachable by $\mathcal{D}$ that satisfy (B3).
We apply the standard saturation algorithm for well-structured transition systems where constraints are propagated from the target control state backwards along the edges of the nMCS.
After its termination, the algorithm computed the minimal preconditions for reaching a the target state.
On a reversed \sNDA, this can be understood as a forward propagation computing minimal post-conditions. 
In this case the target state is $q$ and the minimal post conditions characterize the minimal configurations $(q,f_1,…,f_k)$ that can be reached.
Here, minimal means with the smallest number of instances of some class automaton.
The post-conditions hence give us all minimal sets $M$ when reaching $q$.
These are the (finitely many) candidates for (B4) since if none of those satisfies the properties any larger one will not either.

Now, for testing the candidates $M$ to comply (B4) and $T$ to satisfy (B1), the essential idea is to let the shuffle requirements be checked by a ($k-1$)-\sNDA built by modifying the components of $\mathcal{D}$.
Such an automaton is constructed for each $(m,s)\in M$ and each $t\in T$, respectively, and can, by induction, be checked for emptiness.

\subsection{From \NDLTL to \NDA}
\label{ssc:ndltl2nda}
The translation from $\NDLTL^\pm$ to \sNDA and \pNDA, respectively, follows closely the one for \BDLTL in Section \ref{ssc:bdltl2da}.
For an \NDLTL formula over arbitrarily ordered attributes, a word has at every position a tree of attributes with $f$ maximal paths of length of at most $k$.
The first step is to translate this formula to a formula over the linearly order set of attributes $[k]$ and encode each position of such a word by a segment of length $f$, where each position within a segment corresponds to a maximal path in the tree order $(A,\le)$.
This step is crucial as \NDA only navigate according to linearly ordered attributes.

For translating the obtained formula $\varphi$ into an \NDA, the set $AP_\varphi$ of atomic propositions used by $\varphi$ is extended by propositions $p^\psi_j$ and $=^x_j$ for each $-r_\mathrm{max}\le j\le r_\mathrm{max}$ and subformula $\psi$ and attribute $x$ of $\varphi$, where $r_{max}$ denotes the largest (absolut) value used by $\C^{r}_x$ operators.
As before positional formulae can be checked by the base automaton.
Class formulae of the form $\C^r_x \psi$ can be handled by the local automaton corresponding to attribute $x$.
Propositions $=^x_j$ are checked separately for each attribute $x$ by adapting the construction used for Lemma~\ref{lem:ra}.

Now, together with Theorem~\ref{thm:emptiness-nda}, we obtain a decission procedure for $\NDLTL^\pm$.

\begin{theorem}\label{thm:decidable-ndltl}
Satisfiability of $\NDLTL^+$ is decidable over \emph{finite and infinite} data words.
Satisfiability of $\NDLTL^-$ is decidable over \emph{finite} data words.
\end{theorem}

\begin{corollary}\label{cor:undecidable-pnda}
  Emptiness of \pNDA wrt.\ infinite data words is undecidable.
\end{corollary}

\bibliography{references_full}

\newpage
\appendix

\section{Local Navigation in \BDLTL}
\label{apx:bdltl-fragments}

A straight forward lemma that is used implicitly in the constructions is that the classes of \pDA and \sDA are closed under union and intersection.
\begin{lemma}[Closure]
  Suffix- and prefix-closed data automata are closed under union and intersection.
\end{lemma}
\begin{proof}
For the intersection of two \sDA or two \pDA, carry out the usual product construction the base and class automata separately.
An automaton accepting the union can be constructed by letting the base automaton perform a non-deterministic choice of one of automata and output a flag on the first letter indicating that choice.
The class automaton then simulates the class automaton of the data automaton that was chosen.
It is easy to see that these constructions result in an \sDA or \pDA if the two original automata were both an \sDA or a \pDA, respectively.
\end{proof}

\section{Satisfiability of $\BDLTL^\pm$ is \textnormal{\textsc{2ExpSpace}}-complete}
\label{apx:bdltl-sat}

\subsection{\textnormal{\textsc{ExpSpace}}-variants of Data Automata}

For showing that an \sNDA $\mathcal{D}$ can be checked for emptiness wrt.\ infinite words we use Lemma~\ref{lem:sda-properties} for which we provide more detailed proof here.
Recall that $\mathcal{D} = (\mathcal{A},\mathcal{B})$ with $\mathcal{A}=(Q,\Sigma,\Gamma,\delta_\mathcal{A},Q_0,\emptyset,B_\mathcal{A})$ and $\mathcal{B}=(S,\Sigma,\delta,I,F,B)$.

\subsubsection{Lemma~\ref{lem:sda-properties} (necessity).}
Assume $\mathcal{D}$ has an accepting run $\rho\in(Q\times\mathfrak{F})^\omega$ on some word $w\in(\Sigma\times\Delta)^\omega$.
Let $T\subseteq\delta$ be the set of transitions of the class automaton $\mathcal{B}$ taken infinitely often by $\rho$.
Then, there exists some position $i$ on $\rho$ with $\rho_i = (q,f)$ and, in the suffix $\rho_i\rho_{i+1}…$, only transitions from $T$ are taken by (any instance of) the class automaton $\mathcal{B}$.
If there were a transition $t\in T$ violating property (A1), some instance of $\mathcal{B}$ would reject since $t$ is taken eventually.
As this is not the case, property (A1) holds.

Second, we record that the configuration $(q,f)$ meets the requirements of property (A3).
We assumed that only transitions from $T$ are taken after position $i$ carrying $(q,f)$.
Hence, if there were some data value $d\in\Delta$ violating (A3), the corresponding instance of $\mathcal{B}$ would not accept since it can neither take any transition anymore nor is it in an accepting state.

Third, the suffix $\rho_i\rho_{i+1}…$ of $\rho$ is a witness that (A2) is satisfied.

\subsubsection{Lemma~\ref{lem:sda-properties} (sufficiency).}
To see that the opposite direction also holds, consider a run $\rho$ of $\mathcal{D}$ that leads to $(q,f)$ and continues by applying the the sequence of transitions $\tau=t_1t_2…$ provided by (A2).
$\rho$ is accepting since $\mathcal{A}$, starting in $q$, can correctly continue to move and produce the labels of the transitions while meeting its Büchi condition.
Further, each active instance of $\mathcal{B}$ is identified by some data value $d\in\Delta$ with
$f(d)\neq\bot$.
For those with $f(d)\in F$, we can just consider them discontinued, hence they accept.
A crucial point is that we can indeed always apply the transition sequence $\tau$ since even if there is no active instance of $\mathcal{B}$ that allows for a particular transition $t=(s,g,s')$, we can always spawn a new instance of $\mathcal{B}$ in $s\in S$ since $\mathcal{B}$ is \emph{suffix closed}, i.\,e.\, all states are initial.

Property (A1) guarantees, that all instances of $\mathcal{B}$ are accepting when considering the following scheduling approach.
Put all (finitely many) active instances in some queue of „temporarily completed“ instances.
Whenever a new instance is created it is appended to the queue.
Take the first instance of the queue. 
The last transition taken by it, there is a finite sequence of transitions that leads $\mathcal{B}$ to a final or Büchi state.
This sequence is guaranteed to occur as a (possibly scattered) subsequence in $\tau$ and we just wait for the next suitable transition to occur while dispatching all intermediate transitions to other (existing or new) instances of $\mathcal{B}$.
Upon reaching a final or Büchi state we can dismiss it or append it to the queue, respectively.
All instances are scheduled infinitely often or terminate in some final state.
This way we can construct an accepted data word where the data value at each position corresponds to the active instance of $\mathcal{B}$.

\subsection{From $\BDLTL^\pm$ to Data Automata}

\begin{proof}[Simulating \RA (Lemma \ref{lem:ra})]
Intuitively, the \DA $\mathcal{D'}$ that is supposed to be empty if and only if the intersection of the \DA $\mathcal{D}$ and the \RA $\mathcal{R}$ is empty, simulates both simultaneously. 
The base automaton of $\mathcal{D}'$ simulates the base automaton of $\mathcal{D}$ and the finite control of $\mathcal{R}$. 
In its state it keeps track of the registers currently in use. 
A register becomes marked as in use, when a data value is stored. 
Before a new data value can be stored, the register has to be marked as free. Therefore, the base automaton takes a special transition marked by the new input symbol $\$$ and guesses the data value currently stored in the register. The base automaton encodes in its output when it marks a register as in use or as free, and which registers have to be compared for equality or inequality with respect to the current data value. 
The class automaton of $\mathcal{D}'$ simulates the class automaton of $\mathcal{D}$ and keeps track of the registers the data value associated with the class projection it reads is currently stored in. 
Thus it can verify that the current data value is actually stored in the register that the base automaton expects.

For sake of simplicity, we do not consider Büchi accepting states for the following construction. They can be easily handled using the usual product construction for Büchi automata.
Let $\mathcal{D} = (\mathcal{A}, \mathcal{B})$ be a data automaton $\mathcal{A} = (Q_\mathcal{A}, \Sigma, \Gamma, \delta_\mathcal{A}, I_{\mathcal{A}}, F_{\mathcal{A}}, \emptyset)$ and $\mathcal{B} = (Q_\mathcal{B}, \Gamma, \delta_\mathcal{B}, I_{\mathcal{B}}, F_{\mathcal{B}}, \emptyset)$.
Let $\mathcal{R} = (Q_\mathcal{R}, \Sigma, k, \delta_\mathcal{R}, I_{\mathcal{R}}, F_{\mathcal{R}}, \emptyset)$ be a register automaton.
We can define $\mathcal{D}' = (\mathcal{A}', \mathcal{B}')$ with base automaton $\mathcal{A}' = (Q_{\mathcal{A}'}, \Sigma', \Gamma', \delta_{\mathcal{A}'}, I_{\mathcal{A}'}, F_{\mathcal{A}'})$ and class automaton $\mathcal{B}' = (Q_{\mathcal{B}'}, \Gamma', \delta_{\mathcal{B}'}, I_{\mathcal{B}'}, F_{\mathcal{B}'})$ with $\Sigma' = \Sigma \cup \{\$\}$ and $\Gamma' = \Gamma \times 2^{[k]} \times 2^{[k]} \times [k] \cup [k]$ by:
  \begin{itemize}
    \item $Q_{\mathcal{A}'} = Q_{\mathcal{A}} \times Q_{\mathcal{R}} \times 2^{[k]}$
    \item $I_{\mathcal{A}'} = I_{\mathcal{A}} \times I_{\mathcal{R}} \times \{\emptyset\}$
    \item $F_{\mathcal{A}'} = F_{\mathcal{A}} \times F_{\mathcal{R}} \times 2^{[k]}$
    \item $((q_\mathcal{A},q_\mathcal{R},R),a,(\gamma,R_=,R_{\ne},r),(q_\mathcal{A}',q_\mathcal{R}',R\cup\{r\})) \in \delta_{\mathcal{A}'}$ iff:
    \begin{itemize}
      \item $(q_\mathcal{R}, R_=, R_{\ne}, a, r, q_\mathcal{R}') \in \delta_\mathcal{R}$
      \item $(q_\mathcal{A}, a, \gamma, q_\mathcal{A}') \in \delta_\mathcal{A}$
      \item $r \not\in R$
    \end{itemize}
    \item $((q_\mathcal{A}, q_\mathcal{R}, R), \$, r, (q_\mathcal{A}, q_\mathcal{R}, R \setminus \{r\})) \in\delta_{\mathcal{A}'}$ iff $r \in R$
    \item $Q_{\mathcal{B}'} = Q_\mathcal{Q} \times 2^{[k]}$
    \item $I_{\mathcal{B}'} = I_{\mathcal{B}} \times \{\emptyset\}$
    \item $F_{\mathcal{B}'} = F_{\mathcal{B}} \times \{\emptyset\}$
    \item $((q_\mathcal{B},R),(\gamma,R_=,R_{\ne},r),(q_\mathcal{B}',R\cup\{r\})) \in \delta_{\mathcal{B}'}$ iff
    \begin{itemize}
      \item $(q_\mathcal{B},\gamma,q_\mathcal{B}') \in \delta_\mathcal{B}$
      \item $R_= \in R$
      \item $R \cap R_- = \emptyset$
      \item $r \not\in R$
    \end{itemize}
    \item $((q_\mathcal{B},R),(\gamma,R_=,R_{\ne},r),(q_\mathcal{B},R\setminus\{r\})) \in \delta_{\mathcal{B}'}$ iff $r\in R$
  \end{itemize}
  In every step for every non-empty register $r$ only one instance of the class automaton my be in a state, denoting that the associated data value is stored in $r$. This is guaranteed by the construction  by assuring, using the base automaton, that a new data value is only stored once a register as been freed and that it is only freed if it contains a value, and by assuring using the class automaton, that the value at the storing position is the same as at the next following freeing position. It can be observed, that when the class automaton sees a store action, it will eventually see a free action (and no store or free action in between), or vice versa when it sees a free action, it has already seen a store action.
  Hence, the \DA can be turned into a prefix or suffix closed \DA by turning all states either into final (and Büchi accepting) or initial states.
\end{proof}

\section{Ordered Navigation on Multi-attributed Data Words}

\begin{proof}[Tuple Navigation (Proposition \ref{prp:tuple-nav-undecidable})]
  $\BDLTL^+$ subsumes $LRV^\top$ which is known to be undecidable already over finite words when being extended by tuple navigation \cite{DBLP:conf/lics/DemriFP13}.
Let $\varphi$ be a $\BDLTL^+$ formula over finite words. 
We can construct a satisfiability-equivalent formula $\hat{\varphi}$ by replacing all future operators by past operators and vice versa and replacing every subformula $\psi$ by $\neg \$ \wedge \psi$. Now $\F (\hat{\varphi} \wedge \X \G \$)$ is satisfiable if and only if $\varphi$ is. 
Thus, the satisfiability problem of $\BDLTL^-$ tuple navigation is also undecidable.
\end{proof}

\begin{proof}[Encoding rMCS (Lemma \ref{lem:encoding-rmcs})]
  For the set of attributes we take $A :=\lbrace x_c,\hat{x}_c\mid c\in C\rbrace$, i.\,e., two attributes for each counter of $\mathcal{M}$.
  Let the tree order $\le$ be defined s.\,t.\ attributes for different counters are incomparable and  $x_c>\hat{x}_c$ for all $c\in C$.
  We let $\Phi_\mathcal{M}$ be the the conjunction composed the following formulae.
  \begin{itemize}
    \item $\Phi_\delta$ specifies that a data word has to have the shape of a run 
           according to the finite relation $\delta$, e.\,g., the correct order of 
           states, followed by an operator, a counter and, again, a state.           
           This can be done by only using plain LTL formulae.
    \item $\Phi_\mathrm{res} := \bigwedge_{c\in C}\G((\mathrm{res} \land \X c) \to 
          \C^0_{\hat{x}_c}\neg\Y^= \top)$ specifies, that whenever a reset happens on
          any counter $c$, the current data value has never been seen before 
          (in any of the top-level attributes $\hat{x}_{c'}$, including that 
          for $c$).
    \item $\Phi_\mathrm{dec} := \bigwedge_{c\in C}\G((\mathrm{dec} \land \X \mathrm{c}) \to 
            \C^0_{x_c} \Y^=(\mathrm{inc} \land \X c))$ says, that for
            each decrement operation, the previous position with the same data 
            must carrying an increment operation on the same counter.
            Note that in conjunction with $\Phi_\mathrm{res}$, there must not be a 
            reset on the same counter in between since using $x_c$ in the formula
            means comparing the values of both attributes, $x_c$ and $\hat{x}_c$.
      \end{itemize}
\vspace{-\baselineskip}
\end{proof}

\section{Deciding Satisfiability of $\NDLTL^\pm$}\label{apx:ndltl-sat}

\subsection{Nested Data Automata}
\label{spx:nda}

In the following we provide the proof for Lemma~\ref{lem:snda-properties}, stating that it is necessary and sufficient to find a set $T\in\delta_1$, a state $q\in Q$ and a set $M$ of labeled trees such that the conditions (B1)--(B4) hold, in order to decide that an $k$-\sNDA $\mathcal{D}$ to be decide non-empty.

Let $\mathcal{D} = (\mathcal{A},\mathcal{B}_1,…,\mathcal{B}_k)$ with 
\begin{itemize}
  \item $\mathcal{A} = (Q,\Sigma,\Gamma,\delta_\mathcal{A},I_\mathcal{A},\emptyset, B_\mathcal{A})$ and 
  \item $\mathcal{B}_i = (S_i,\Gamma,\delta_i,S_i,F_i, B_i)$ for $i\in[k]$.
\end{itemize}

\begin{proof}[Lemma \ref{lem:snda-properties}, necessity]
Assume $\mathcal{D}$ has an accepting run $\rho\in(Q\times\mathfrak{F}_1\times…\times\mathfrak{F}_k)^\omega$ on some infinite data word $w=w_0w_1…\in(\Sigma\times\Delta^{[k]})^\omega$ with $w_i=(a_i,\vec{d}_i)$.
The run $\rho$ induces a sequence $\tau\in\delta_1^\omega$ of transitions of $\mathcal{B}_1$ that are performed between consecutive configurations.
This need not to be a run in $\mathcal{B}_1$ but it is by definition a shuffle of runs of $\mathcal{B}_1$ and since $\rho$ is accepting each of these \emph{class} runs is accepting.
Let $T\subseteq\delta_1$ be the maximal set of transitions of $\mathcal{B}_1$ that occur infinitely often in $\tau$.

There is a position $N$ in $\tau$ (and $\rho$) after which only transitions from $T$ occur.
Any transition $t_1\in T$ occurs eventually on the suffix of $\tau$ starting at this position.
The class run that $t_1$ belongs to is accepting and since all states in $\mathcal{B}_1$ are initial also each suffix of that class run is an accepting run of $\mathcal{B}_1$. 
In particular, the suffix starting with $t_1$. 

This is our witness for (B1): It consists only of transitions from $T$. 
Further, by the definition runs of \NDA and the fact that all class automata states are initial, the sequence $\gamma$ of labels along that run of $\mathcal{B}_1$ is a shuffle of words accepted by $\mathcal{B}_2$ which in turn are shuffles of words accepted by $\mathcal{B}_3$ and so on.
Intuitively, the instances of $\mathcal{B}_2$ reading parts of $\gamma$ are those where the corresponding data valuation is an extension of the data valuation corresponding the instance of $\mathcal{B}_1$ reading $\gamma$.

We choose the state $q$ to be that occurring in the configuration $\rho_N = (q,f_1,…,f_k)$ at position $N$ in $\rho$.
Recall, that after $N$ only transitions from $T$ occur in $\tau$.
The the suffix $\tau_{N+1}\tau_{N+2}…$ of $\tau$ is then the witness for (B2). 
Since $\rho$ is accepting, $\mathcal{A}$, starting in $q$ can output the sequence of labels corresponding to $\tau_{N+1}\tau_{N+2}…$.

Third, we record that $\rho_N = (q, f_1,…,f_k)$ is a reachable configuration for which property (B3) holds.
Property (B4) states that at the configuration $\rho_N$ fixed for property (B3), each active instance of a class automaton $\mathcal{B}_x$ and all instances depending on it (present and spawned at later positions) must accept when continuing with $\rho$.
This is the case since $\rho$ is accepting.
\end{proof}

\begin{proof}[Lemma \ref{lem:snda-properties}, sufficiency]
We have seen that if an \sNDA $\mathcal{D}$ is non-empty, there is a reachable configuration $c = (q,f_1,…,f_k)$ inducing a set $M_c$ of trees and a set of transitions $T$ such that properties (1)-(4) hold.

To see that the opposite direction also holds, consider a run $\rho$ of some \sNDA $\mathcal{D} = (\mathcal{A},\mathcal{B}_1,…,\mathcal{B}_k)$ that leads to the configuration $c = (q,f_1,…,f_k)$ and continues by applying the the sequence of transitions $\tau\in T^\omega$ provided by (B2).
The base automaton $\mathcal{A}$ accepts since starting in $q$ it can correctly continue to move and produce the labels of the transitions while meeting its Büchi condition.
This run of $\mathcal{A}$ already yields the string projection for data word $w$ accepted by $\mathcal{D}$ and it remains to argue that it is possible to correctly choose data values for $w$.

Up to reaching configuration $c$, property (B3) assures that there is a choice of data values such that all class projections are accepted by the corresponding class automaton, except for those represented in $M_c$.
By (B4) we can “complete” these instances as follows.
For each root of a tree $(m,s)\in M_c$ we know that there is a sequence of transitions $\tau_{(m,s)}\in T^\infty$ of $\mathcal{B}_1$ such that $\mathcal{B}_1$ accepts. 
Further, that run induces a word $\gamma_{(m,s)}\in\Gamma^\infty$ and it is possible to decompose it into subsequences that can be assigned to the depending active instances (represented by the other nodes in $m$) or to newly spawned instances of class automata of the respective level such that all those instances accept.
Note that a spawning a new instance on some level $x$, which corresponds to a fresh data value in the model, implies spawning a new dependent instance for each level $y>x$ as the data words we consider do not allow for missing data values.
Property (B4) assures acceptance of those as the words annotated to leaf nodes that are not on the bottom level must still be a shuffle of words accepted by the class automata $\mathcal{B}_{y'}$, $k\ge y'>y$, below.

We “execute” each sequence $\tau_{(m,s)}=t_1t_2…$ by choosing the first occurrence of $t_1$ in the global sequence $\tau$ and devote it to the instance represented by the root of $(m,s)$ by labeling the position of model $w$ where $t_1$ occurs with the corresponding data value.
We continuously choose the respective next occurrence of the transitions $t_2…$ and proceed analogously.
This is possible since (B2) guarantees that in $\tau$ each transition occurs always eventually.

Having treated these “obligations” imposed by $M_c$, we can now fill the remaining positions where data values are missing in $w$ as follows.
Choose fresh data values for the first free position.
For the transition $t\in T$ taken at this position in $\tau$, property (B1) provides a sequence assuring that the new instances spawned will accept.
This sequence is distributed over the global sequence $\tau$ the same way as the sequences $\tau_{(m,s)}$ above.
Note that we can always choose a fresh data value, no matter which transition needs to be taken since the corresponding class automaton is suffix-closed and can thus be started in any of its states.
The model $w$ constructed that way is a witness that $\mathcal{D}$ is non-empty.
\end{proof}

\subsection{From \NDLTL to \NDA}

For translating \NDLTL formulae to \NDA we first reduce formulae over arbitrary tree-orders $(A,\le)$ to a satisfiability-equivalent formula over the \emph{linearly} ordered set $[k]$.
We provide this translation in detail in terms of the following lemma.

\begin{lemma}
Let $(A,\le)$ be a tree order of attributes with tree-height $k$.
For each $\NDLTL^\pm$ formula over $(A,\le)$ we can construct a satisfiability-equivalent $\NDLTL^\pm$ formula over attributes $[k]$ with the linear order on the natural numbers.
\end{lemma}
\begin{proof}
Let $f$ be the number of maximal paths in the graph of $(A,\le)$.
We enumerate them and let $f(x)$, for $x\in A$, denote the lowest number $i$ such that $x$ occurs on the $i$-th such path.
Further, let $0\le h(x)<k$ be the level of $x$ in the corresponding tree, i.\,e.\, the distance to its root.
We let $A':= \lbrace y_0,…,y_{k-1}\rbrace$ be linearly ordered by $y_0\le'y_1\le'…\le'y_{k-1}$.

The idea is now, similar to the construction in Section \ref{ssc:bdltl2da}, to encode words over the tree order $(A,\le)$ into words over the linear order of attributes $[k]$, where we have segments of length $f$ corresponding to single positions in the original data word. 
Each position within a certain frame corresponds to a maximal path in the tree order $(A,\le)$.

We transform a given $\NDLTL^\pm$ formula $\Phi$ over attributes $(A,\le)$ and propositions $AP$ into an $\NDLTL^\pm$ formula $\hat{\Phi}$ over attributes $[k]$ and propositions $AP' := AP \cupdot\lbrace p_1,…,p_f\rbrace$.
The additional propositions are intended to mark the positions in each frame (i.\,e.\, modulo $f$), which can easily be expressed by an \LTL formula $\Phi_e$.

We can assume that $\Phi$ has a normal form where each $\X^=$, $\Y^=$, $\U^=$ and $\S^=$ formula is directly preceded by $\C^r_x$.
This is due to the equalities 
\begin{align*}
   \C^r_x\neg\varphi &\equiv (\neg\C^r_x \varphi) \land \beta_r,\\
   \C^r_x(\varphi\land\psi) &\equiv (\C^r_x \varphi) \land (\C^r_x \psi),\\
   \C^r_x\X\varphi &\equiv \X^r\varphi,\\
   \C^r_x\Y\varphi &\equiv \Y^r\varphi,\\
   \C^r_x(\varphi\U\psi) &\equiv \X^r(\varphi\U\psi),\\
   \C^r_x(\varphi\S\psi) &\equiv \X^r(\varphi\S\psi),\\
   \varphi\U^=\psi &\equiv \psi \lor (\varphi \land \X^= ( \alpha(\varphi) \U^= \alpha(\psi) )), \\
   \varphi\S^=\psi &\equiv \psi \lor (\varphi \land \Y^= ( \alpha(\varphi) \S^= \alpha(\psi) )), \\
   \X^=\varphi  &\equiv \X^=\alpha(\varphi), \text{ and} \\
   \Y^=\varphi  &\equiv \Y^=\alpha(\varphi),
\end{align*}
where $\beta_r$ is an \LTL formula checking that it is possible to move $r$ steps, $\alpha(\varphi) := \bigwedge_{x\in A}( @x \to \C^0_x \varphi )$ and, for $r<0$, $\X^r$ denotes $\Y^{-r}$ and vice versa.

Then, $\hat{\Phi} := \Phi_e \land t(\Phi)$ where we define $t(\Phi)$ as follows.
\begin{align*}
  t(p)   &= p \\
  t(\neg\varphi)  &= \neg t(\varphi) \\
  t(\varphi\land\psi) &= t(\varphi) \land t(\psi) \\
  t(\varphi\lor\psi) &= t(\varphi) \lor t(\psi) \\
  t(\X \varphi) &= \X^f t(\varphi) \\
  t(\Y\varphi) &= \Y^f t(\varphi) \\
  t(\varphi\U \psi) &= t(\varphi) \U t(\psi) \\
  t(\varphi\S \psi) &= t(\varphi) \S t(\psi) \\
  t(\C^r_x (\varphi \U^= \psi)) &= \bigwedge_{j=1}^f p_j \to \X^{f(x)-j}\C^{fr}_{y_{h(x)}} (\varphi \U^= \psi)\\
  t(\C^r_x (\varphi \S^= \psi)) &= \bigwedge_{j=1}^f p_j \to \X^{f(x)-j}\C^{fr}_{y_{h(x)}} (\varphi \S^= \psi)\\
  t(\C^r_x \X^= \varphi) &= \bigwedge_{j=1}^f p_j \to \X^{f(x)-j}\C^{fr-f(x)+f}_{y_{h(x)}} (\X^=\varphi)\\
  t(\C^r_x \Y^= \varphi) &= \bigwedge_{j=1}^f p_j \to \X^{f(x)-j}\C^{fr-f(x)}_{y_{h(x)}} (\Y^=\varphi)
\end{align*}

The correctness of the transformation can easily be seen by considering the underlying invariant, that the constructed formulae are evaluated equally on every position of a particular frame.
Note that under the transformation, any $\NDLTL^\pm$ formula stays in its respective fragment.
\end{proof}

For the final translation to \NDA we rely on the fact that we can check the additional propositions $=_j^x$.
This can be easily done using a register automaton. 
To obtain the type of automaton we need, such as \pNDA or \sNDA, we can then use the following lemma.

\begin{lemma}
For the linearly ordered set of attributes $[k]$ and a natural number $r\in\mathbb{N}$, let $\Sigma$ be a finite alphabet that includes a set $P = \lbrace=_j^x\mid-r\le j\le r, x\in A\rbrace\subseteq\Sigma$ of dedicated propositions.

Let $\mathcal{D}$ be a $k$-nested \NDA over $A$, $\Sigma$ and a data domain $\Delta$ and 
$E\subseteq(\Sigma\times\Delta^{[k]})^\infty$ be the language of all $[k]$-attributed data words $w = \binom{a_0a_1…}{\vec{d}_0\vec{d}_1…}$, $a_i\in\Sigma$, $\vec{d}_i\in\Delta^{[k]}$, such that
\[
  \forall_{0\le i<|w|} =_j^x\in a_i \text{ iff } \vec{d}_i|_{[x]} = \vec{d}_{i+j}|_{[x]}.
\]
We can construct a $k$-nested \NDA $\mathcal{D}'$ such that $\mathcal{D}' = \emptyset$ iff $\mathcal{D}\cap E = \emptyset$.
Further, if $\mathcal{D}$ is an \sNDA or an \pNDA then $\mathcal{D}$ is an \sNDA or an \pNDA, respectively.
\end{lemma}

\begin{proof}
For each attribute $x$ we can, in the same fashion as in Section \ref{ssc:bdltl2da}, build a dedicated register automaton that reads tuples of data values $(d_1,…,d_k)$ and checks the correct annotation of propositions $=_j^x$ wrt.\ to the values in the \mbox{$x$-th} level. 
Following the construction from Lemma \ref{lem:ra}, we can construct an \NDA that simultaneously simulates $\mathcal{D}$ and all of the register automata for each level $x$.
For each register automaton, the construction only involves the base automaton and the respective $x$-th class automaton.
As argued above, the single class automata remain suffix- or prefix-closed.
\end{proof}

\end{document}